\documentclass[letterpaper, 10 pt, conference]{ieeeconf} 
\IEEEoverridecommandlockouts                             
\usepackage{graphicx} 

\usepackage{setspace}

\usepackage{amsmath, amsthm, amssymb}
\usepackage{siunitx}
\usepackage{hyperref}
\usepackage{mathtools}
\usepackage{comment}
\usepackage{acro}
\usepackage{todonotes}

\DeclareAcronym{DARE}{
  short=DARE,
  long=Deep-space Autonomous Robotic Explorer,
}

\DeclareAcronym{MuSCAT}{
  short=MuSCAT,
  long=Multi-Spacecraft Concept and Autonomy Tool,
}

\DeclareAcronym{EKF}{
  short=EKF,
  long=Extended Kalman Filter,
}

\DeclareAcronym{NEOs}{
short = NEOs, 
long = Near-Earth Objects
}

\DeclareAcronym{OREX}{
short = OSIRIS-REx,
long = {Origins, Spectral Interpretation, Resource Identification, and Security-Regolith Explorer}
}

\DeclareAcronym{SoC}{
  short=SoC,
  long= State-of-Charge, 
}

\DeclareAcronym{SRP}{
  short=SRP,
  long=Solar Radiation Pressure,
}

\DeclareAcronym{TCMs}{
    short = TCMs, 
    long = Trajectory Correction Maneuvers,
}

\DeclareAcronym{BBF}{
    short = BBF, 
    long = Bennu Body-fixed Frame,
}

\DeclareAcronym{GNC}{
    short = GNC, 
    long = {Guidance, Navigation and Control},
}

\DeclareAcronym{CVaR}{
    short = CVaR, 
    long = Conditional Value at Risk,
}

\DeclareAcronym{VaR}{
    short = VaR, 
    long = Value at Risk,
}

\DeclareAcronym{ADCS}{
    short = ADCS, 
    long = Attitude Determination and Control System,
}

\DeclareAcronym{MIP}{
    short = MIP, 
    long = Mixed-Integer Programming,
}

\DeclareAcronym{ODM}{
    short = ODM, 
    long = Orbit Departure Maneuver,
}

\DeclareAcronym{ORM}{
    short = ORM, 
    long = Orbit Recapture Maneuver,
}

\DeclareAcronym{HARP}{
    short = HARP,
    long = Hierarchical Asteroid Reconnaissance Planner,
}

\DeclareAcronym{QQ}{   
    short = QQ,
    long = Quantile-Quantile,
}

\DeclareAcronym{CADRE}{
 short = CADRE, 
 long = Cooperative Autonomous Distributed Robotic Exploration, 
}

\DeclareAcronym{SCP}{
 short = SCP, 
 long = Sequential Convex Programming, 
}

\DeclareAcronym{DAREabst}{
  short=DARE,
  long=Deep-space Autonomous Robotic Explorer,
}

\DeclareAcronym{MuSCATabst}{
  short=MuSCAT,
  long=Multi-Spacecraft Concept and Autonomy Tool,
}

\DeclareAcronym{SA}{
    short = SA,
    long = Sampling Approximation,
}

\DeclareAcronym{TOF}{
    short = TOF, 
    long = Time of Flight,
}

\DeclareAcronym{SPA}{
    short = SPA,
    long =  Symmetric Polytopic Approximation,
}

\DeclareAcronym{HARPabst}{
    short = HARP,
    long = Hierarchical Asteroid Reconnaissance Planner,
}

\DeclareAcronym{SOCP}{
    short = SOCP,
    long = Second Order Cone Programming,
}

\DeclareAcronym{SDP}{
    short = SDP,
    long = Semi Definite Programming,
}

\DeclareAcronym{LROE}{
    short = LROE,
    long = modified linearized relative orbital element,
}

\DeclareAcronym{LTV}{
    short = LTV, 
    long =  linear time-varying,
}

\DeclareAcronym{LTI}{
    short = LTI,
    long = linear time-invariant,
}

\DeclareAcronym{CW}{
    short = CW, 
    long = Clohessy-Wiltshire,
}

\DeclareAcronym{EP}{
    short = EP,
    long = electric propulsion,
}

\DeclareAcronym{NEXT}{
    short = NEXT,
    long = NASA’s Evolutionary Xenon Thruster,
}

\DeclareAcronym{SOAV}{
    short = SOAV,
    long = sum-of-absolute-values,
}

\DeclareAcronym{IPM}{
    short = IPM,
    long = interior-point method,
}

\DeclareAcronym{LP}{
    short = LP,
    long = Linear Programming,
}

\DeclareAcronym{RCS}{
    short = RCS,
    long = reaction control system,
}

\DeclareAcronym{PTR}{
    short = PTR,
    long = penalized trust region,
}

\DeclareAcronym{3-DOF}{
    short = 3-DOF,
    long = three-degree-of-freedom,
}

\DeclareAcronym{SAA}{
    short = SAA,
    long = sample average approximation,
}

\DeclareAcronym{SDE}{
    short = SDE,
    long = stochastic differential equation,
}

\DeclareAcronym{PDF}{
    short = PDF,
    long = probability density function,
}

\DeclareAcronym{FOH}{
    short = FOH, 
    long = first-order hold,
} 
\DeclareAcronym{WC-CVaR}{
 short = WC-CVaR,
 long = Worst Case CVaR,
 }

 \DeclareAcronym{WC-VaR}{
    short = WC-VaR,
    long = Worst Case VaR,
 }

 \DeclareAcronym{WC-CC}{
    short = WC-CC,
    long = Worst Case Chance Constraint,
 }

\DeclareAcronym{InSAR}{
    short = InSAR,
    long = Interferometric Synthetic Aperture Radar,
}

\DeclareAcronym{JSC}{
  short = JSC,
  long = Johnson Space Center,
}

\DeclareAcronym{CBF}{
  short = CBF,
  long = Control Barrier Function,
}

\DeclareAcronym{QP}{
  short = QP,
  long = Quadratic Program,
}

\DeclareAcronym{DKW}{
    short = DKW, 
    long = Dvoretzky--Kiefer--Wolfowitz,
}

\DeclareAcronym{CDF}{
    short = CDF,
    long = Cumulative Distribution Function,
}

\DeclareAcronym{SOS}{
    short = SOS,
    long = Sum of Squares
}

\usepackage{lipsum}

\makeatletter
\newcommand{\blankfootnote}[1]{%
  \bgroup
  \renewcommand{\thefootnote}{}%
  \footnotetext{\refstepcounter{footnote}#1}%
  \egroup
}

\newtheorem{theorem}{Theorem}
\newtheorem{lemma}{Lemma}

\newtheorem{corollary}{Corollary}
\theoremstyle{definition}

\newtheorem{remark}{Remark}
\newtheorem{assumption}{Assumption}
\newtheorem{problem}{Problem}

\newcommand{\R}{\mathbb{R}}
\newcommand{\N}{\mathcal{N}}

\newcommand{\norm}[1]{\left\|#1\right\|}
\newcommand{\inner}[2]{\langle #1, #2 \rangle}

\DeclareMathOperator{\CVaR}{CVaR}

\makeatother
\begin{document}

\title{\bf Probabilistic Control Barrier Functions for Systems with State Estimation Uncertainty using Sub-Gaussian Concentration}

\author{Kazuya Echigo, David E. J. van Wijk, Pol Mestres, Ersin Da\c{s}, Joel W. Burdick, and Aaron D. Ames \, 
\thanks{The work was supported by Air Force Office of Scientific Research Grant No. 113535-19668. The research was carried out in part at the Jet Propulsion Laboratory, California Institute of Technology, under a contract with the National Aeronautics and Space Administration (80NM0018D0004).}
\thanks{KE, DVW, PM, JB, and AA are with the Division of Engineering and Applied Science, California Institute of Technology, Pasadena, CA 91125, USA (e-mail: kazuyae@caltech.edu). E. Da\c{s} is with the Department of Mechanical, Materials, and Aerospace Engineering, Illinois Institute of Technology, Chicago, IL 60616, USA.}
}

\maketitle
\begin{abstract}

Safety-critical control systems, such as spacecraft performing proximity operations, must provide formal safety guarantees despite stochastic uncertainties from state estimation and unmodeled dynamics. 
Although~\acp{CBF} have been extended to stochastic systems, existing approaches typically face a trade-off between the tightness of probabilistic guarantees and computational tractability.
This paper presents a particle-based probabilistic~\ac{CBF} framework that overcomes this limitation by exploiting the sub-Gaussian structure of the barrier function increment under Gaussian uncertainties. We establish that Gaussian uncertainties propagating through Lipschitz-continuous control-affine dynamics preserve sub-Gaussianity of the barrier function increment, with explicit tail bounds. Leveraging this structure, we derive finite-sample bounds on the approximation error between particle-based~\ac{CVaR} estimates and ground-truth probabilistic constraints; applying this yields a tractable optimization problem formulation with finite-sample safety certificates. We show through numerical experiments how the proposed approach provides tight yet provably valid probabilistic safety guarantees.

\end{abstract}

\section{Introduction}

Designing safe controllers using~\acp{CBF} has been an extensive area of study within the control community.~\ac{CBF}-based controllers can be synthesized with provably safe guarantees and low computational overhead, with effectiveness validated in multiple applications~\cite{Ames2017}. However, conventional~\acp{CBF} have rarely been applied to systems subject to both model and state estimation uncertainty—such as deep space explorations~\cite{NesnasHockmanEtAl2021,kazuDARE2025}. This limitation arises because they are fundamentally designed for deterministic systems or systems with bounded disturbances. When the system is subject to stochastic disturbances with unbounded support—as commonly arises when using state estimators—such deterministic guarantees fail: the system will almost surely violate the safety constraints in finite time~\cite{Ryan2023_RSS}.

Motivated by this challenge, considerable research has focused on stochastic extensions of~\acp{CBF}, which provide probabilistic safety guarantees~\cite{clark_control_2019,mestres2025}. However, a key remaining challenge is obtaining a tractable representation of state estimation uncertainty as it propagates through nonlinear affine dynamics—the typical dynamical model in~\ac{CBF} applications~\cite{Ames2019}. When state uncertainty propagates through such dynamics, the resulting distribution generally lacks a closed-form expression. This forces existing methods to either adopt overly conservative approximations or rely on computationally expensive representations.

With that in mind, we focus our attention on extending probabilistic~\acp{CBF} to deal with both state estimation and model uncertainty, and formulate a tractable optimization-based algorithm to compute a probabilistically safe control input. Specifically, we exploit the sub-Gaussian structure of the barrier function increment under Gaussian uncertainties to derive tight, explicit concentration bounds without requiring conservative assumptions. This yields a tractable optimization problem (a~\ac{QP} when the safety and input constraints are linear) with rigorous probabilistic guarantees suitable for real-time implementation.

\paragraph{Related Works} 

There has been substantial effort to adapt~\ac{CBF} techniques for systems with uncertainty. Initial work focused on robust methods for bounded, deterministic disturbances~\cite{skolathaya2019_input,aanil2022_safe,mjankovic2018_robust,wang2023_disturbance}. Building on these foundations while relaxing those assumptions, other works have explored the problem under stochastic and possibly unbounded disturbances using modeling techniques such as martingale-based analysis~\cite{csantoyo2019_verification} or scenario sampling~\cite{mestres2025,aadonascimiento2024_probabilistically}. While these approaches account for uncertainty propagation through system dynamics, they either do not account for state estimation uncertainty or require restrictive assumptions.

Some works have examined this problem by specializing to specific classes of estimators. In continuous time,~\cite{aclark2021_control,yaghoubi_risk-bounded_2021} use martingale arguments to generate probabilistic safety certificates, while in discrete time,~\cite{jsteinhardt2012_finite} apply similar martingale-based approaches. Alternatively,~\cite{vahs2023belief,kishida2024risk} specialize to the extended Kalman filter in discrete time. However, those approaches only leverage the first moment of the uncertainty distribution and can be conservative when more distributional information is available. A promising technique that has recently emerged is to leverage the scenario-based approach using distribution-free concentration inequalities such as~\ac{DKW} inequality~\cite{Matti2024}. While this work accommodates state estimation uncertainty in nonlinear control-affine systems, 
it requires the random variable to have bounded support: the bound must be known a priori, and for unbounded distributions, artificial truncation is needed, introducing an approximation error that does not decay with the number of samples.

In contrast to these existing approaches, we propose a particle-based~\ac{CBF} framework that exploits the sub-Gaussian structure of the barrier function increment. 
Our approach requires no distribution truncation or bounded support assumptions, yielding an approximation error that decays to zero as the number of samples increases at a rate of $O((n\ln n)^{-1/2})$.

\paragraph{Statement of Contributions}
We develop a particle-based probabilistic~\ac{CBF} framework that simultaneously handles both state estimation and dynamics uncertainty with rigorous finite-sample safety guarantees. Our technical contributions are threefold:
\begin{enumerate}
    \item We establish that the barrier function increment is sub-Gaussian 
    when Gaussian uncertainties propagate through Lipschitz-continuous 
    nonlinear control-affine dynamics, with an explicit parameter bound.
    \item We derive explicit, finite-sample bounds on the particle-based~\ac{CVaR} approximation error without requiring the uncertainty to have bounded support. 
        \item We reformulate the probabilistic~\ac{CBF} condition as a tractable optimization problem, resolving the theory-computation tradeoff in existing approaches.
\end{enumerate}

The remainder of this paper is organized as follows. Section~\ref{sec:prelim} introduces the problem formulation and reviews probabilistic~\acp{CBF} and risk measures. Section~\ref{sec:subgaussian} presents our approach: establishing sub-Gaussianity of the barrier function increment, deriving finite-sample~\ac{CVaR} bounds, and reformulating the probabilistic~\ac{CBF} condition as a tractable optimization problem. Section~\ref{sec:numerical} demonstrates the efficacy of our approach through numerical experiments.

\section{Problem Setup} \label{sec:prelim}

We begin with a discrete-time stochastic system\footnote{
We use $\norm{\cdot}$ for the Euclidean norm and $\norm{\cdot}_{\mathrm{op}}$ for the operator norm.
For a symmetric matrix $\Sigma$, we denote its largest and smallest eigenvalues by $\lambda_{\max}(\Sigma)$ and $\lambda_{\min}(\Sigma)$. For any $a \in \mathbb{R}$, the positive part is $(a)_+ := \max(0, a)$. Also, let $[n] := \{1, 2, \ldots, n\}$ be the index set.}
\begin{equation}\label{eq:dynamics}
x_{t+1} = f(x_t) + g(x_t)u_t + d_t := F(x_t, u_t, d_t),
\end{equation}
where $x_t \in \R^{n_x}$ is the unknown true state, $u_t \in \R^{n_u}$ is the control input, and $d_t \in \R^{n_x}$ is an additive, state-independent process disturbance signal. The functions ${f: \R^{n_x} \to \R^{n_x}}$ and ${g: \R^{n_x} \to \R^{n_x \times n_u}}$ are Lipschitz continuous with constants $L_f$ and $L_g$, respectively. 

\begin{assumption}\label{ass:uncertainty}
The true state $x_t \sim \N(\mu_{x_t}, \Sigma_{x_t})$ and disturbance $d_t \sim \N(\mu_{d_t}, \Sigma_{d_t})$ are independent Gaussian random vectors with known moments. We assume the disturbance distribution is state and time-independent. The state estimate $\mu_{x_t}$ and uncertainty covariance $\Sigma_{x_t}$ may come from an estimator such as an~\ac{EKF}.
\end{assumption}

Let $h: \R^{n_x} \to \R$ be a barrier function defining the safe set as the zero sublevel set (opposite to standard~\acp{CBF}~\cite{Ames2017}):
\begin{equation*}
\mathcal{C} = \{x \in \R^{n_x} : h(x) \leq 0\}.
\end{equation*}

\begin{assumption}\label{ass:barrier}
The barrier function $h$ is Lipschitz continuous with constant $L_h$ and is convex. 
\end{assumption}

Our objective is to design an output feedback controller $k: \R^{n_x} \to \R^{n_u}$ that keeps the system trajectories within the safe set $\mathcal{C}$ at all times. Since the true state $x_t$ is unobservable, the controller must operate using the state estimate and account for its uncertainty. In deterministic systems, the discrete-time~\ac{CBF} condition $h(x_{t+1}) \leq \gamma h(x_t)$ with $\gamma \in [0,1]$ ensures forward invariance of $\mathcal{C}$~\cite{Ames2017,agrawal2017discrete}. However, stochastic disturbances fundamentally change this picture. When disturbances have unbounded support, the system exits the safe set in finite time with probability one~\cite{Ryan2023_RSS}, making deterministic safety guarantees impossible. This necessitates a probabilistic framework for safety. Following~\cite{mestres2025}, we say a controller $k: \R^{n_x} \to \R^{n_u}$ is $\epsilon$-safe over horizon $H$ if, starting from any $x_0 \in \mathcal{C}$, the probability of remaining in $\mathcal{C}$ for all $t \in [0,H]$ is at least $1-\epsilon$. To construct such controllers, we leverage the notion of probabilistic~\ac{CBF} developed in~\cite{mestres2025}. The key idea is to enforce a single-step probabilistic constraint: for each $x_t \in \mathcal{C}$, there exists a $u_t \in \mathcal{U}$ satisfying
\begin{equation}\label{eq:prob_cbf_condition}
\mathbb{P}\big(h(F(x_t,u_t,d_t)) \leq \gamma h(x_t)\big) \geq 1 - \alpha,
\end{equation}
where $\alpha \in (0,1)$ is the single-step risk parameter and $\gamma \in [0,1]$ is a decay rate. It has been shown~\cite{mestres2025} that if \eqref{eq:prob_cbf_condition} is satisfied at each time step with $\alpha \leq 1 - (1-\epsilon)^{1/H}$, the controller achieves $\epsilon$-safety over horizon $H$.

With this theoretical foundation for safe control, we implement an optimization problem with a probabilistic~\ac{CBF}, which is solved at each time step. Given a desired control $u_{\text{des}} \in \R^{n_u}$ from a nominal policy, we seek to solve
\begin{equation}\label{eq:main_problem}
\begin{aligned}
\min_{u_t \in \mathcal{U}} \quad & \norm{u_t - u_{\text{des}}}^2 \\
 \text{s.t. } \quad & \mathbb{P}\big(h(F(x_t, u_t, d_t)) \leq \gamma h(x_t)\big) \geq 1 - \alpha 
\end{aligned}
\end{equation}
where $\mathcal{U} \subseteq \R^{n_u}$ is a compact admissible control set with $u_{\max} := \sup_{u \in \mathcal{U}} \|u\|$.
In practice, this stochastic optimization problem presents several computational challenges. The chance constraint is in general non-convex, and evaluating the probability requires the distribution of
\begin{equation} \label{eq:main_rv}
\Delta h(x_t, u_t, d_t) := h(F(x_t, u_t, d_t)) - \gamma h(x_t),
\end{equation}
which is usually intractable for nonlinear $f$, $g$, and $h$. Note that we henceforth suppress the arguments of ${\Delta h}$ 
whenever there is no ambiguity.

\section{Proposed Approach}

We address the intractability of~\eqref{eq:main_problem} through two key steps: reformulating the constraint using~\ac{CVaR} and estimating it through particle propagation. Recall that~\ac{CVaR} provides a sufficient condition for chance constraints:
\footnote{$\CVaR_\alpha (\xi) = \mathbb{E} [ \xi |  \xi  \geq q_\alpha ]$ represents the expected value of a random variable $\xi$ in its worst $1-\alpha$-percent of outcomes, where $q_\alpha$ is the $\alpha$-quantile.}
\begin{align}
\CVaR_\alpha (\Delta h) \leq 0 
\implies  \mathbb{P}(\Delta h \leq 0) \geq 1 - \alpha.\label{eq:cvar_sufficient}
\end{align}
This reformulation replaces the discontinuous chance constraint with a continuous sufficient condition, amenable to tractable empirical approximations.
Evaluating $\CVaR_\alpha(\Delta h)$ requires knowledge of the distribution of $\Delta h$. Under Assumption~\ref{ass:uncertainty}, we represent the joint uncertainty in state and disturbance using $n$ i.i.d. samples:
\begin{align*}
\{(x_t^i, d_t^i)\}_{i=1}^n, \,\,
\text{where } x_t^i \sim \N(\mu_{x_t}, \Sigma_{x_t}), \,\, d_t^i \sim \N(\mu_{d_t}, \Sigma_{d_t}).    
\end{align*}
Accordingly, for a given control input $u_t$, each particle propagates through the dynamics as
\begin{equation}\label{eq:particle_dynamics}
x_{t+1}^i = f(x_t^i) + g(x_t^i)u_t + d_t^i, \quad i \in [n]
\end{equation}
yielding barrier function increments
\begin{equation}\label{eq:particle_barrier}
\Delta h^i(u_t) = h(x_{t+1}^i) - \gamma h(x_t^i), \quad i \in [n].
\end{equation}
These particle-based realizations, $\{\Delta h^i(u_t)\}_{i=1}^n$, 
yield a tractable constraint upper-bounding the true~\ac{CVaR}.

\subsection{Sub-Gaussianity of the Barrier Function Increment}
\label{sec:subgaussian}

To bound the gap between the particle-based~\ac{CVaR} approximation and the ground truth, we exploit the distributional properties of the barrier function increment $\Delta h $. In this section, we establish that $\Delta h$ is sub-Gaussian with an explicit parameter that depends on Lipschitz constants and covariance matrices~\footnote{A random vector $X$ is \emph{sub-Gaussian with parameter $\sigma$} if, for all $\lambda \in \mathbb{R}$ and unit vectors $\|v\|=1$, $\mathbb{E}[\exp(\lambda\langle X - \mathbb{E}[X], v \rangle)] \leq e^{\lambda^2 \sigma^2 / 2}$.}.
We begin with standard results.

\begin{lemma}[sub-Gaussian Vectors]\label{lem:gaussian_vector}
If $X \sim \N(\mu, \Sigma)$, then $X$ is sub-Gaussian with parameter $\sigma = C\sqrt{\lambda_{\max}(\Sigma)}$, where $C$ is a universal constant (typically $C \leq \sqrt{2}$).
\end{lemma}

\begin{proof}
For any unit vector $v$, $\inner{X - \mu}{v} \sim \N(0, v^\top \Sigma v)$. Since $v^\top \Sigma v \leq \lambda_{\max}(\Sigma)$, the result follows from standard Gaussian concentration.
\end{proof}

\begin{lemma}[Multivariate Lipschitz Composition]\label{lem:lipschitz_comp_multi}
Let $\phi: \R^n \to \R^m$ be Lipschitz continuous with constant $L$, and let $\xi \sim \N(\mu, \Sigma)$ be a Gaussian random vector in $\R^n$. Then $\phi(\xi)$ is sub-Gaussian with parameter $\sigma \leq CL\sqrt{\lambda_{\max}(\Sigma)}$, where $C$ is a universal constant as in Lemma~\ref{lem:gaussian_vector}.
\end{lemma}

\begin{proof}
For any unit vector $v \in \R^m$, consider $\psi(x) = \inner{\phi(x)}{v}$. This is a Lipschitz function from $\R^n$ to $\R$ with constant $L$ (since $|\psi(x_1) - \psi(x_2)| \leq \norm{\phi(x_1) - \phi(x_2)} \leq L\norm{x_1 - x_2}$). By standard Lipschitz-Gaussian concentration, $\psi(\xi)$ is sub-Gaussian with parameter $L\sqrt{\lambda_{\max}(\Sigma)}$. Since this holds for all unit vectors $v$, the result follows.
\end{proof}

These lemmas follow from standard Gaussian concentration results; see~\cite{vershynin2018high} for details. We now prove $\Delta h$ follows a sub-Gaussian distribution and its parameter is bounded. 

\begin{theorem}[Sub-Gaussianity of $\Delta h$]\label{thm:subgaussian_multi}
The random variable $\Delta h$ defined in \eqref{eq:main_rv} is sub-Gaussian with parameter
\begin{equation*}
\sigma_{\Delta h} \leq C\sqrt{L_{\Phi,x}^2\lambda_{\max}(\Sigma_x) + L_{\Phi,d}^2\lambda_{\max}(\Sigma_d)},
\end{equation*}
where $L_{\Phi,x} = L_h(L_f + L_g u_{\max} + |\gamma|)$ and $L_{\Phi,d} = L_h$.
\end{theorem}
\begin{proof}
Consider ${\Delta h \!=\! h(f(x) \!+\! g(x)u \!+\! d) \!-\! \gamma h(x)}$. Note that ${h(f(x) + g(x)u + d))}$ and $\gamma h(x)$ are \textit{not independent} (both depend on $x$), so we cannot simply use the sum of independent sub-Gaussians property. Define $\Phi: \R^{n_x} \times \R^{n_x} \to \R$ by $\Phi(x, d) = h(f(x) + g(x)u + d) - \gamma h(x)$, so that $\Delta h = \Phi(x,d)$. A straightforward calculation using the Lipschitz properties of $h$, $f$, and $g$ shows that
\begin{equation*}
|\Phi(x_1, d_1) - \Phi(x_2, d_2)| \leq L_{\Phi,x}\norm{x_1 - x_2} + L_{\Phi,d}\norm{d_1 - d_2}.
\end{equation*}
Lemma~\ref{lem:lipschitz_comp_multi} and the joint Gaussianity of $(x, d)$ then yield the stated bound.
\end{proof}

\subsection{Finite-Sample~\ac{CVaR} Approximation}
\label{sec:cvar_bound}
Having established the sub-Gaussian property of $\Delta h$, we now derive finite-sample bounds that relate the particle-based~\ac{CVaR} to the true~\ac{CVaR}. This enables replacing the intractable~\ac{CVaR} constraint in~\eqref{eq:main_problem} with a computable sample-based condition. We begin with a technical lemma bounding tail integrals for sub-Gaussian random variables.

\begin{lemma}[Sub-Gaussian Tail Integral Bound] \label{lem:tail_integral_bound}
Let $X$ be a sub-Gaussian random variable with mean $\mu$ and parameter $\bar{\sigma}^2$. Let $\epsilon \in (0, 0.5)$ be a probability threshold. Suppose $z$ is a threshold such that $\mathbb{P}(X > z) \leq \epsilon$. Then, the tail integral is bounded as:
\begin{equation*}
\mathbb{E}\left[(X - z)\mathbf{1}_{\{X > z\}}\right] = \int_z^{\infty} \mathbb{P}(X > y) \, dy \leq \frac{\bar{\sigma} \epsilon}{\sqrt{2\ln(1/\epsilon)}}.
\end{equation*}
\end{lemma}
\begin{proof}
Let $g(y) := \exp\left(-\frac{(y - \mu)^2}{2\bar{\sigma}^2}\right)$. By the sub-Gaussian assumption, we have:
\begin{equation*}
\int_z^{\infty} \mathbb{P}(X > y) \, dy \leq \int_z^{\infty} g(y) \, dy.
\end{equation*}

To obtain the worst-case bound, we consider the threshold $z^*$ that saturates the constraint:
\begin{align*}
g(z^*) = \epsilon \iff z^* = \mu + \bar{\sigma}\sqrt{2\ln(1/\epsilon)}.
\end{align*}
Since $\epsilon < 0.5$, we have $z^* > \mu$ and the integral is maximized at $z = z^*$. With the change of variable $t = (y - \mu)/\bar{\sigma}$, where $t^* = \sqrt{2\ln(1/\epsilon)} > 0$:
\begin{equation*}
\int_{z^*}^{\infty} g(y) \, dy = \bar{\sigma} \int_{t^*}^{\infty} e^{-t^2/2} \, dt.
\end{equation*}
Applying Mills' inequality ${\int_t^{\infty} \! e^{-u^2/2} \, du \!\leq\! \frac{1}{t}e^{-t^2/2}}$ for ${t \!>\! 0}$:
\begin{equation*}
\bar{\sigma} \int_{t^*}^{\infty} e^{-t^2/2} \, dt \leq \bar{\sigma} \cdot \frac{e^{-(t^*)^2/2}}{t^*} = \bar{\sigma} \cdot \frac{\epsilon}{\sqrt{2\ln(1/\epsilon)}},
\end{equation*}
where we used $e^{-(t^*)^2/2} = \epsilon$.
\end{proof}

With this tail bound established, we can now state our main result: a finite-sample upper bound on~\ac{CVaR} expressed in terms of order statistics. While the theorem involves order statistics for the analysis, the practical implementation in Section~\ref{sec:tractable} avoids computing them explicitly.

\begin{theorem}[Data-Dependent Sub-Gaussian~\ac{CVaR} Bound]\label{thm:optimal_subgaussian_cvar_mubar}
Let $X_1, \ldots, X_n$ be i.i.d.\ sub-Gaussian random variables with mean $\mu$ and sub-Gaussian parameter $\sigma^2 > 0$. 
Let $Z_1 \leq \cdots \leq Z_n$ denote the order statistics of $X_1, \ldots, X_n$. Then for any $\alpha \in (0,1)$ and $\delta \in (0, 0.5]$ with $\varepsilon_n(\delta) < 0.5$\footnote{This condition is satisfied when $n > 2\ln(2/\delta)$.},
\begin{align*}
&\mathbb{P}\bigg[\CVaR_\alpha(X) \leq Z_n + C_{\emph{tail}}(n,\delta) - \\
&\quad \frac{1}{\alpha}\sum_{i=1}^{n-1}(Z_{i+1} - Z_i)\left(\frac{i}{n} - \varepsilon_n(\delta) - (1-\alpha)\right)^+\bigg] \geq 1-\delta,
\end{align*}
where $\varepsilon_n(\delta) := \sqrt{\frac{\ln(2/\delta)}{2n}}$ and $C_{\emph{tail}}(n, \delta) := 
\frac{\bar{\sigma}\, \varepsilon_n(\delta)}{\alpha\sqrt{2\ln(1/\varepsilon_n(\delta))}}$
\end{theorem}

\begin{proof}
Let $F$ and $F_\omega$ denote the true and empirical~\acp{CDF}, respectively. By the~\ac{DKW} inequality with Massart's tight constants~\cite{Massart10},
\begin{equation*}
\mathbb{P}\left(\sup_{x \in \mathbb{R}} (F_\omega(x) - F(x)) \leq \varepsilon_n(\delta) \right) \geq 1 - \delta.
\end{equation*}
Define the event $\mathcal{E}_{\text{DKW}} := \{\sup_{x}(F_\omega(x) - F(x)) \leq \varepsilon_n(\delta)\}$. On this event, since $F_\omega(Z_n) = 1$, we have $F(Z_n) \geq 1 - \varepsilon_n(\delta)$, and hence $\mathbb{P}(X > Z_n) \leq \varepsilon_n(\delta)$. Applying Lemma~\ref{lem:tail_integral_bound} with $z = Z_n$ then yields
\begin{equation*}
I_{\text{tail}} := \int_{Z_n}^{\infty} (1 - F(y))\, dy \leq \frac{\bar{\sigma}\, \varepsilon_n(\delta)}{\sqrt{2\ln(1/\varepsilon_n(\delta))}} =: \bar{I}_{\text{tail}}.
\end{equation*}

Now define the lower bound~\ac{CDF} $G_\omega^-(x) := \max\{0, F_\omega(x) - \varepsilon_n(\delta)\}$, which satisfies $G_\omega^-(x) \leq F(x)$ for all $x$ on $\mathcal{E}_{\text{DKW}}$. Using the integral representation of~\ac{CVaR}:
\begin{align*}
\CVaR_\alpha(X) \leq \inf_{c \in \mathbb{R}}\left\{c + \frac{1}{\alpha}\left(\int_c^{Z_n}(1 - G_\omega^-(y))dy + I_{\text{tail}}\right)\right\}.
\end{align*}
Applying Lemma~1 in~\cite{Thomas2019}, this yields
\begin{equation*}
\CVaR_\alpha(X) \leq Z_n - \frac{1}{\alpha}\int_{-\infty}^{Z_n}(G_\omega^-(y) - (1-\alpha))^+ dy + \frac{I_{\text{tail}}}{\alpha}.
\end{equation*}
$F_\omega$ is piecewise constant with ${F_\omega(x) \!= \!i/n, \,\forall\, x \!\in \!(Z_i, Z_{i+1}]}$, and so ${G_\omega^-(x) \!= \!\max\{0, i/n \!-\! \varepsilon_n(\delta)\}}$ on each interval. Thus:
\begin{align*}
&\int_{-\infty}^{Z_n}(G_\omega^-(y) - (1-\alpha))^+ dy \\ 
&= \sum_{i=1}^{n-1}(Z_{i+1} - Z_i)\left(\frac{i}{n} - \varepsilon_n(\delta) - (1-\alpha)\right)^+ .
\end{align*}
Setting $C_{\text{tail}}(n,\delta) := \bar{I}_{\text{tail}}/\alpha$ completes the proof.
\end{proof}

\begin{remark} $\bar{I}_{\emph{tail}}$ decays as $O((n\ln n)^{-1/2})$.
\end{remark}

\begin{remark}[A Posteriori Verification]\label{rem:a-posteriori-verification}
Theorem~\ref{thm:optimal_subgaussian_cvar_mubar} assumes i.i.d.\ samples, which may be violated when the same particles are used for both optimization and evaluation.
In our experiments, we evaluate with new samples. Although no verification failures were observed in our case, standard remedies include resampling or tightening the risk level~\cite{zhao2025}.
\end{remark}

Alternatively, when the barrier function increment admits a separable structure, the i.i.d.\ assumption holds directly without requiring sample splitting as detailed in Remark~\ref{rem:a-posteriori-verification}.

\begin{corollary}[Separable Structure]\label{cor:separable}
Suppose $\Delta h = a(u_t) + b(x_t, d_t)$ for some functions $a$ and $b$. Let $\{(x_t^i, d_t^i)\}_{i=1}^n$ be i.i.d.\ samples and define $Z_1 \leq \cdots \leq Z_n$ as the order statistics of $\{b(x_t^i, d_t^i)\}_{i=1}^n$. 
Then any $u_t$ satisfying
\begin{eqnarray*}
    a(u_t) + Z_n + C_{\emph{tail}}(n,\delta) - \\
\frac{1}{\alpha}\sum_{i=1}^{n-1}(Z_{i+1} - Z_i)\left(\frac{i}{n} - \varepsilon_n(\delta) - (1-\alpha)\right)^+ \leq 0,
\end{eqnarray*}
guarantees $\CVaR_\alpha(\Delta h) \leq 0$ with probability at least $1-\delta$, by the translation invariance of $\CVaR$.
\end{corollary}
    
We note that for the control-affine dynamics~\eqref{eq:dynamics}, this separable structure holds when $h$ is linear and $g$ is constant. It arises naturally in linear systems or fully-actuated mechanical systems with position-based safety constraints.

\begin{remark}[Safety Over a Finite Time Horizon Provided by Theorem~\ref{thm:optimal_subgaussian_cvar_mubar}]
Suppose Theorem~\ref{thm:optimal_subgaussian_cvar_mubar} holds 
with confidence $1-\delta$ at each $t \in [H]$ using independent particles.
Let $q_0 := \mathbb{P}(h(x_0) \leq 0)$ denote the probability that 
the true initial state lies in $\mathcal{C}$ under 
Assumption~\ref{ass:uncertainty}. Then, setting $\alpha \leq \epsilon/H$, 
an analogue of Proposition~6 of~\cite{mestres2025} 
holds under state uncertainty:
\begin{equation}
    \mathbb{P}_{\mathrm{ptcl}}\!\left(
    \mathbb{P}\!\left(\bigcap_{t=0}^{H}\{h(x_t)\leq 0\}\right)
    \geq q_0 - \epsilon
    \right) \geq 1 - H\delta,
\end{equation}
where $\mathbb{P}_{\mathrm{ptcl}}$ denotes the probability 
over particle randomness. 
\end{remark}

\subsection{\ac{CVaR} Reformulation as a Linear Program}
\label{sec:tractable}

Theorem~\ref{thm:optimal_subgaussian_cvar_mubar} provides a high-probability upper bound on $\CVaR_\alpha(\Delta h)$ involving order statistics and correction terms. While theoretically powerful, this form is not immediately suitable for optimization: it requires explicit sorting and the bound's structure is not amenable to standard solvers. We resolve this by reformulating the bound as a linear program that avoids explicit computation of order statistics. We recall a classical result that expresses sample~\ac{CVaR} as an optimization problem over unordered samples~\cite{rockafellar2000optimization}.

\begin{theorem}[Sample~\ac{CVaR} Formula~\cite{rockafellar2000optimization}]
\label{thm:cvar_original}
Let $\{w_1, \ldots, w_n\}$ be $n$ real numbers with order statistics $w_{(1)} \leq \cdots \leq w_{(n)}$, and let $\alpha \in (0,1)$. Then
\begin{align*}    
&\inf_{\theta \in \mathbb{R}} \left\{ \theta + \frac{1}{n\alpha} \sum_{i=1}^{n} (w_i - \theta)_+ \right\} \\
&= w_{(n)} - \frac{1}{\alpha} \sum_{j=1}^{n-1} (w_{(j+1)} - w_{(j)}) \left( \frac{j}{n} - (1-\alpha) \right)_+.
\end{align*}
\end{theorem}

Crucially, the left-hand side depends only on the unordered samples $\{w_i\}$ and can be computed via optimization, while the right-hand side involves order statistics. This equivalence enables computation without explicit sorting. To accommodate the correction term $\varepsilon_n(\delta)$ from Theorem~\ref{thm:optimal_subgaussian_cvar_mubar}, we extend this classical formula. Let $B \in (0, \alpha)$ be a shift in the risk threshold. Then we have the following result.

\begin{corollary}[Shifted Threshold~\ac{CVaR} Formula]
\label{cor:cvar_shifted}
Under the setting of Theorem~\ref{thm:cvar_original}, let $B \in (0, \alpha)$. Then
\begin{align*}
\frac{B}{\alpha} w_{(n)} + \frac{\alpha - B}{\alpha} \inf_{\theta \in \mathbb{R}} \left\{ \theta + \frac{1}{n(\alpha - B)} \sum_{i=1}^{n} (w_i - \theta)_+ \right\} \\ 
= w_{(n)} - \frac{1}{\alpha} \sum_{j=1}^{n-1} (w_{(j+1)} - w_{(j)}) \left( \frac{j}{n} - B - (1-\alpha) \right)_+.    
\end{align*}
\end{corollary}

\begin{proof}
Factor $w_{(n)}$ from the right-hand side and apply Theorem~\ref{thm:cvar_original} with risk level $\alpha - B$.
\end{proof}

Setting $B = \varepsilon_n(\delta)$ yields the tractable constraint for our probabilistic~\ac{CBF} formulation. 
The optimization problems in Theorem~\ref{thm:cvar_original} and Corollary~\ref{cor:cvar_shifted} allow efficient linear programming reformulations. The constraint
\begin{equation}\label{eq:cvar_constraint_abstract}
\frac{B}{\alpha} w_{(n)} + \frac{\alpha - B}{\alpha} \inf_{\theta} \left\{ \theta + \frac{1}{n(\alpha - B)} \sum_{i=1}^{n} (w_i - \theta)_+ \right\} \leq 0,
\end{equation}
holds iff there exist $\theta \in \mathbb{R}$ and $\{s_i \geq 0\}_{i=1}^n$ satisfying
\begin{flalign}\label{eq:lp_constraint}
\begin{aligned}
\frac{B}{\alpha} w_i \!+\! \frac{\alpha-B}{\alpha} \left( \theta \!+\! \frac{1}{n(\alpha-B)} \sum_{j=1}^n s_j \right) \leq 0, \, i \in [n], \\
s_i \geq w_i \!-\! \theta, \quad s_i \geq 0, \, i \in [n].
\end{aligned}
\end{flalign}

The constraint must hold for all samples including the maximum, yielding the first inequality. The auxiliary variables $s_i$ capture the positive parts $(w_i - \theta)_+$. 

\subsection{Tractable Probabilistic~\ac{CBF} Optimization}
\label{sec:algorithm}

We now apply the reformulation from Section~\ref{sec:tractable} to construct a computationally efficient version of~\eqref{eq:main_problem}. Theorem~\ref{thm:optimal_subgaussian_cvar_mubar} guarantees with probability at least $1-\delta$ that
\begin{align*}    
&\CVaR_\alpha(\Delta h) \leq \Delta h_{(n)} + C_{\text{tail}}(n,\delta) \\
&- \frac{1}{\alpha} \sum_{j=1}^{n-1} (\Delta h_{(j+1)} - \Delta h_{(j)}) \left( \frac{j}{n} - \varepsilon_n(\delta) - (1-\alpha) \right)_+,
\end{align*}
where $\{\Delta h^i\}_{i=1}^n$ are particle realizations and $\Delta h_{(1)} \leq \cdots \leq \Delta h_{(n)}$ denote their order statistics. 
Applying Corollary~\ref{cor:cvar_shifted} using ${w_i\! =\! \Delta h^i}$ and ${B \!=\! \varepsilon_n(\delta)}$, with the reformulation~\eqref{eq:lp_constraint}, the sufficient condition $\CVaR_\alpha(\Delta h) \leq 0$ is equivalent to the existence of $\theta \in \mathbb{R}$ and $\{s_i \geq 0\}_{i=1}^n$ satisfying
\begin{align}
    &C_{\text{tail}}(n,\delta) + \frac{\varepsilon_n(\delta)}{\alpha} \Delta h^i (u_t) + \frac{\alpha - \varepsilon_n(\delta)}{\alpha} \bigg( \theta \nonumber \\ &+ \frac{1}{n(\alpha - \varepsilon_n(\delta))} \sum_{j=1}^n s_j \bigg) \leq 0, \quad i \in [n], \label{eq:const1} \\
&s_i \geq \Delta h^i (u_t) - \theta, \quad s_i \geq 0, \quad i \in [n].\label{eq:const2}
\end{align}
Combining this~\ac{CVaR} constraint with the control objective and particle dynamics yields the following optimization problem, which serves as a computationally tractable approximation to~\eqref{eq:main_problem}.

\begin{problem}\label{prob:stochastic_cbf_qp}
Given particles $\{(x_t^i, d_t^i)\}_{i=1}^n$ sampled according to Assumption~\ref{ass:uncertainty} and a desired control $u_{\text{des}} \in \R^{n_u}$ from a nominal controller, solve
\begin{equation}\label{eq:final_qp}
\begin{aligned}
\min_{u_t \in \mathcal{U} , \theta, \{s_i\}} \quad & \|u_t - u_{\text{des}}\|^2 \\
\text{s.t. } \quad 
& \text{\eqref{eq:const1}, \ \eqref{eq:const2}.}   
\end{aligned}
\end{equation}
\end{problem}
We remark that, when $
\Delta h^i(u_t)$ is convex in $u_t$ and $\mathcal{U}$ is a convex set, Problem~\ref{prob:stochastic_cbf_qp} is a convex quadratic program. Moreover, when they are linear, the problem is a~\ac{QP}.

\begin{remark}
By applying Corollary~\ref{cor:cvar_shifted} to $\{b(x_t^i, d_t^i)\}_{i=1}^n$, the constraint in Corollary~\ref{cor:separable} can be reformulated as a linear program in auxiliary variables $(\theta, \{s_i\})$. When $a(u_t)$ is affine in $u_t$, the resulting problem remains a~\ac{QP} with the same computational complexity as Problem~\ref{prob:stochastic_cbf_qp}. 
\end{remark}

\section{Numerical Experiments}\label{sec:numerical}

We demonstrate the efficacy of our proposed approach through Monte Carlo simulations on a non-holonomic mobile robot navigating under uncertainty\footnote{Implementation available at: \url{https://github.com/kazuyae/Stochastic_CBF}}. The robot must reach a goal while respecting a geofence constraint, subject to both localization uncertainty and stochastic dynamics. Our approach is compared against two alternatives: a deterministic~\ac{CBF} baseline ignoring uncertainty and a probabilistic~\ac{CBF} using the conservative~\ac{DKW} inequality.

The robot operates in a 2D workspace with state ${x \!=\! [r_x, r_y, \theta]^\top \!\in\! \mathbb{R}^3}$, where $(r_x, r_y)$ is the position and $\theta$ is the heading angle. The control input is ${u \!=\! [v, \omega]^\top}$, where $v$ and $\omega$ are linear and angular velocities with ${|v| \! \leq\! v_{\rm m}}$, ${|\omega| \!\leq \!\omega_{\rm m}}$, where $v_{\rm m}$, $\omega_{\rm m}$ are derived from an Ackermann robot with a wheelbase of ${w_l}$ and a maximum steering angle of $\eta$. To obtain a formulation suitable for~\ac{CBF} design, we follow~\cite{Paul2019} and introduce a shifted reference point. Define
\begin{equation}
    R(\theta) = \begin{bmatrix} \cos\theta & -\sin\theta \\ \sin\theta & \cos\theta \end{bmatrix}, \quad p = s + \ell R(\theta)e_1,
\end{equation}
where $s = [r_x, r_y]^\top$ is the robot's center position, $e_1 = [1, 0]^\top$, and ${\ell > 0}$ is a design parameter. This shifted point $p$ lies orthogonal to the wheel axis and satisfies $\dot{p} = R(\theta)L^{-1}u$, where ${L = \text{diag}(1, 1/\ell)}$. 
Thus, both linear and angular velocity appear in the dynamics of $p$, yielding relative degree one for position-based barrier functions.

The barrier function $h(x)$ is a linear geofence: ${h(x) = r_y}$. The robot experiences two sources of uncertainty: Gaussian localization noise and Gaussian process disturbances with covariances $\Sigma_x,\Sigma_d$, respectively. A nominal PID controller drives the robot toward the goal at $(0, -0.05)$, while our~\ac{CBF}-based safety filter modifies the control input when necessary. In our simulation, the state and its time-varying uncertainty covariance are dynamically estimated via an~\ac{EKF} with process noise $Q = \Sigma_d$, measurement noise $R$, and initial covariance $P_0$, using partial measurements of $r_y$ and $\theta$.

We perform a Monte Carlo analysis with $10,000$ trials\footnote{The simulations use the constants 
${P_0 = \text{diag}([0.02, 0.02, 0.07]^2)}$, ${\Sigma_d \!=\! Q \!=\! \text{diag}([0.01, 0.01, 0.05]^2)}$, $R = \text{diag}([0.02, 0.07]^2)$,
$\mu_0 = [0, -0.5, \pi/2]^\top$, $v_{\rm m} = \SI{0.3}{m/s}$, $\omega_{\rm m} = \SI{0.67}{rad/s}$, $w_l = 2.5 \,{\rm m}$, $n = 500$, $\eta = 40^\circ$, $\alpha = 0.1$, $\delta = 0.1$, $\gamma = 0.2$, $\Delta t = \SI{0.5}{s}$.}. Each trial begins at $x_0 \sim \mathcal{N}(\mu_0, P_0)$
and
runs for up to $\SI{15}{s}$ with a time step of $\Delta t$.
A trial is deemed successful if the robot reaches within $\SI{0.02}{m}$ of the goal.  

Table~\ref{tab:monte_carlo_results} summarizes all trials and Figure~\ref{fig:trajectories} shows representative trajectories from $10$ randomly selected trials for each method. The deterministic~\ac{CBF} frequently violates the safety constraint due to its inability to account for state and process uncertainties. In contrast, the~\ac{DKW}-based approach exhibits excessive conservatism, causing many trajectories to stall before reaching the goal. 
The proposed method strikes a balance; the $2.0\%$ violation stays within the $\alpha=0.1$ bound ($10\%$), maintaining safety and consistent goal-reaching.

\begin{table}[t!]
    \centering
    \small
    \begin{tabular}{lcc}
    \hline
    \textbf{Method} & \textbf{Violation Rate} & \textbf{Reached} \\
    \hline
    Deterministic CBF & 100\% & 86\% \\
    Probabilistic CBF (DKW) & 0\% & 56\% \\
    Proposed (Sub-Gaussian) & 2.0\% & 100\% \\
    \hline
    \end{tabular}
    \caption{The summary of Monte Carlo simulations}
    \label{tab:monte_carlo_results}
     \vspace{-.5cm}
\end{table}

\begin{figure}[t!]
    \centering
    \includegraphics[width=0.95\linewidth]{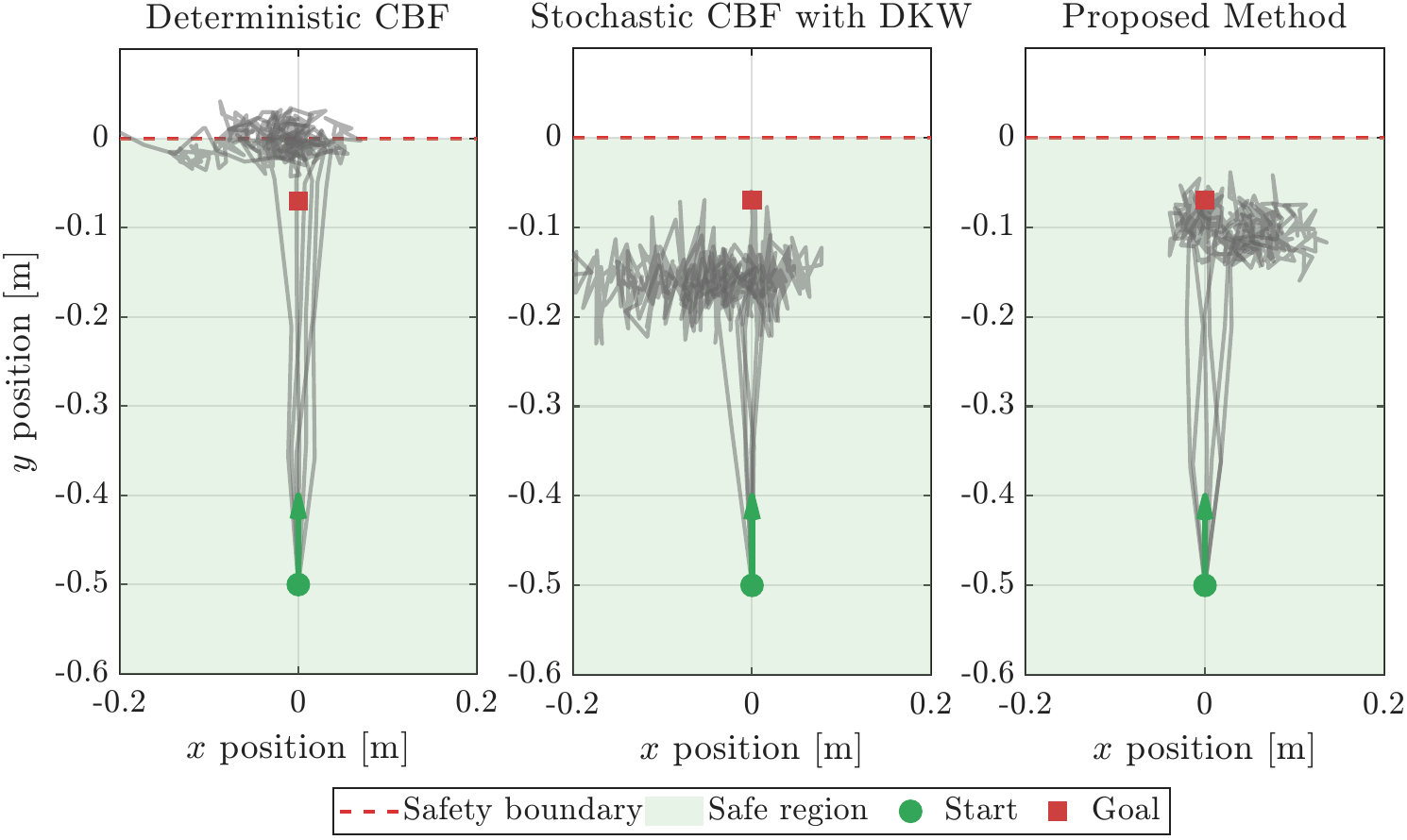} 
    \vspace{-.15cm}
    \caption{Representative trajectories for the three tested methods. The plots display the robot's path from the start position (green circle) to the goal position (red box), with the safety boundary at $y = 0$ (red dotted line). }
    \label{fig:trajectories}
    \vspace*{-0.2cm}
\end{figure}

\begin{figure}[t!]
    \centering
    \includegraphics[width=0.95\linewidth]{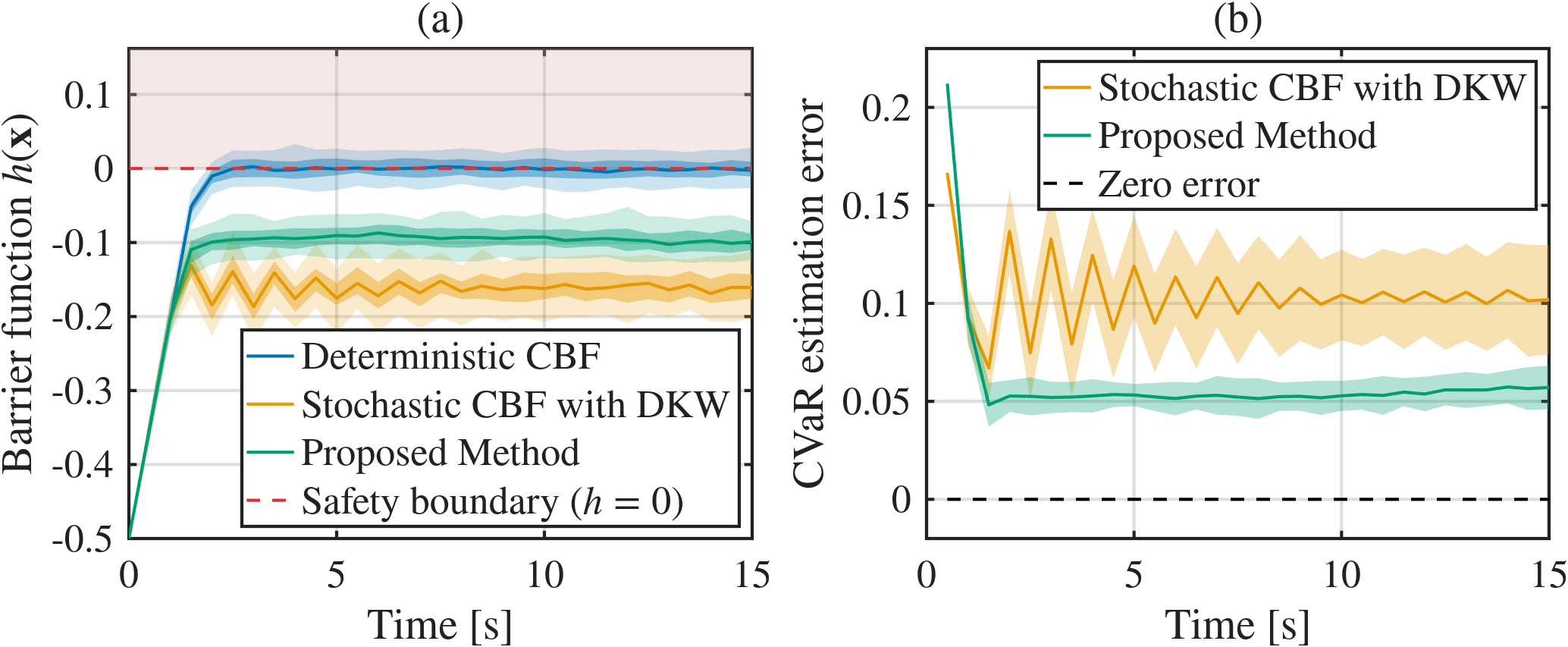}
    \vspace{-.15cm}
    \caption{Performance metrics across Monte Carlo trials. (a) Barrier function evolution over time for each method, with the safety boundary at $h = 0$ (red dotted line). (b)~\ac{CVaR} estimation error over time.} 
    \label{fig:monte_carlo_results}
    \vspace*{-0.38cm}
\end{figure}

Figure~\ref{fig:monte_carlo_results} provides additional insights into the performance of each method. Note that the~\ac{CVaR} estimation error was measured as the difference between the particle-based~\ac{CVaR} estimate when running probabilistic~\ac{CBF}-based methods and a Monte Carlo ground truth computed with $10^6$ samples. We again see that the~\ac{DKW} bound consistently overestimates the true risk, while the proposed method tracks the true~\ac{CVaR} more closely, resulting in tighter approximations.

\section{Conclusion}\label{sec:conclusion}
We presented a particle-based probabilistic~\ac{CBF} framework that provides probabilistic safety guarantees for systems subject to both state estimation uncertainty and process disturbances. Our approach exploits the sub-Gaussian structure of the barrier function increment to derive explicit concentration bounds, enabling the reformulation of the intractable probabilistic~\ac{CBF} condition as a tractable optimization program. Monte Carlo simulations on a non-holonomic mobile robot navigation problem demonstrated that the proposed approach outperforms both the deterministic baseline and the existing truncation-based probabilistic~\ac{CBF} method by providing tighter yet provably valid risk approximations. A current limitation is the double probability structure---the~\ac{CVaR} bound holds with probability $1-\delta$ while the safety constraint holds with probability $1-\alpha$---and unifying them remains an important direction for future work.

\begin{spacing}{0.925}
\bibliographystyle{IEEEtran}
\bibliography{asl_shortened, referencelist}
\end{spacing}

\end{document}